\documentclass[11pt]{elsarticle}
\usepackage{amsmath,amssymb,amsthm}
\usepackage[a4paper]{geometry}
\usepackage{hyperref}
\usepackage[final]{changes}

\newcommand{\Ahat}{\widehat{A}}
\newcommand{\alphah}{\widehat{\alpha}}
\newcommand{\Cc}{\mathbb{C}}
\newcommand{\e}[1]{\times 10^{#1}}
\newcommand{\eps}{\varepsilon}
\newcommand{\gh}{\widehat{g}}
\newcommand{\Hh}{\underline{H}}
\newcommand{\Hht}{\underline{\tilde{H}}}
\renewcommand{\Im}{\mathop{\mathrm{Im}}}
\newcommand{\Kk}{\mathcal{K}}
\newcommand{\Rr}{\mathbb{R}}
\newcommand{\Span}{\mathop{\mathrm{span}}}
\newcommand{\yh}{\widehat{y}}
\newcommand{\yt}{\widetilde{y}}
\newcommand{\zstat}{z_{\mathrm{stat}}}

\newtheorem{coroll}{Corollary}
\newtheorem{lemma}{Lemma}

\begin{document}

\begin{frontmatter}
  
\title{Exponential Krylov time integration\\
for modeling multi-frequency optical response\\ 
with monochromatic sources}

\author[kiam,sk]{M.A.~Botchev\corref{cor1}}
\ead{botchev@kiam.ru}

\author[scho,ut1]{A.M.~Hanse}
\ead{abelhanse@gmail.com}

\author[ut2]{R.~Uppu}
\ead{r.uppu@utwente.nl}

\cortext[cor1]{Corresponding author.  Work of this author is supported by 
the Russian Science Foundation Grant 17-11-01376.}

\address[kiam]{Keldysh Institute of Applied Mathematics,
Russian Academy of Sciences, Miusskaya~Sq.~4, 125047 Moscow,
Russia}

\address[sk]{Skolkovo Institute of Science and Technology, Skolkovo Innovation Center, Bldg~3,
143026 Moscow, Russia}

\address[scho]{International School Twente, Het Stedelijk Lyceum, 
               Tiemeister~20, 7541WG Enschede, the Netherlands}

\address[ut1]{Department of Applied Mathematics and MESA+ Institute for 
Nanotechnology, University of Twente, P.O.~Box 217, 7500~AE
Enschede, the Netherlands}

\address[ut2]{Complex Photonic Systems group (COPS) and MESA+ Institute for 
Nanotechnology, University of Twente, P.O.~Box 217, 7500~AE
Enschede, the Netherlands}


\begin{abstract}
Light incident on a layer of scattering material such as a piece of 
sugar or white paper forms a characteristic speckle pattern in 
transmission and reflection. The information hidden in the correlations 
of the speckle pattern with varying frequency, polarization and angle 
of the incident light can be exploited for applications such as 
biomedical imaging and high-resolution microscopy. Conventional 
computational models for multi-frequency optical response involve 
multiple solution runs of Maxwell's equations with monochromatic sources.  
Exponential Krylov subspace time solvers are promising candidates for 
improving efficiency of such models, as single monochromatic solution 
can be reused for the other frequencies without performing full 
time-domain computations at each frequency. However, we show that 
the straightforward implementation appears to have serious limitations. 
We further propose alternative ways for efficient solution through 
Krylov subspace methods. Our methods are based on two different 
splittings of the unknown solution into different parts, each 
of which can be computed efficiently.  Experiments demonstrate 
a significant gain in computation time with respect to the 
standard solvers.
\end{abstract}

\end{frontmatter}

\section{Introduction}
Modeling of light propagation is an
important and ever-growing research area \cite{Chew2001,Taflove}. Understanding 
the propagation of light in a scattering medium has widespread 
applications such as real-time medical imaging,
spectrometry, and quantum security \cite{bertolotti2012non, katz2014non, 
yilmaz2015speckle, redding2013compact, goorden2014quantum}.
A complex scattering medium comprises of a regular or an irregular spatial 
arrangement of inhomogeneities in refractive index. While the former are 
complex fabricated structures known as photonic crystals, we experience the 
latter in common materials such as skin, sugar, and white paint.
Light propagating through a layer of such random scattering media
undergoes multiple scattering off the inhomogeneities resulting in a complex 
interference pattern, called the 
speckle pattern \cite{akkermans2007mesoscopic}.
The seemingly random speckle patterns
possess rich correlations which depend on parameters such as the frequency, 
angle of incidence and spectral content of the propagating light 
\cite{freund1988memory, van1999multiple, dogariu2015electromagnetic, 
Gigan2016}. These correlations help in revealing fundamental 
light transport properties of the medium which are instrumental
for different applications.

An important modeling approach to analyze frequency correlations of speckle 
pattern is through the so-called broadband pulse excitation (BPE).
Mathematically, this modeling approach solves the time-dependent
Maxwell equations with a source function (representing
the incident light) taken as a short Gaussian pulse in time. The frequency 
bandwidth of the pulse is inversely related to the temporal width of the pulse, 
therefore shorter the pulse, broader the frequency bandwidth. The speckle
patterns at different frequencies can then the obtained by taking a 
Fourier-transform of the optical response (which
is the resulting electromagnetic field after time $T$ of computation when the 
incident pulse has decayed).
A drawback of the BPE approach is that the accuracy of the results is subject 
to the time-window of computation, which is necessarily short for broadband 
light.  The separation of single frequency response from the total response 
will be compromised by the frequency resolution of Fourier-transformed fields. 
In fact, the BPE can be seen a compromise between accuracy 
(increasing if the frequency bandwidth gets narrower) and
efficiency (increasing if the frequency bandwidth gets broader).
To avoid this potential loss in accuracy in BPE simulations or to verify the 
BPE results, time-domain computations
with pulses of a single frequency have to be carried out.

This paper explores a way to utilize the single frequency time-domain 
computations for efficiently computing the multi-frequency response with the 
help of the Krylov subspace exponential time integration. These methods are 
based on a projection (onto the Krylov subspace) which is carried out 
independently of the frequency in the source term. Hence, the same projection 
is used to obtain a small-dimensional projected problem for optical response at 
a different frequency. Therefore, the Krylov subspace methods should be 
computationally more efficient in comparison to multiple single frequency 
computations. Moreover, the projection methods achieve the computational 
efficiency without compromising the accuracy unlike BPE. 

Exponential time integration has received much attention in the recent
decades. These methods involve actions of the matrix functions,
such as matrix exponential and matrix cosine. The efficiency of the methods
depends on their implementation.  
For large matrices, methods to compute matrix function actions
on a vector include Krylov subspace methods (based on Lanczos or Arnoldi
processes), Chebyshev polynomials (primarily for Hermitian
matrices),
and scaling and squaring with Pad\'e or Taylor 
approximations and some other 
methods~\cite{TalEzer89,DeRaedt03,SchmelzerTrefethen07,CaliariOstermann09,%
AlmohyHigham2011}.
In this paper we use Krylov subspace methods for computing
the matrix exponential actions on vectors; these methods
are suitable for non-Hermitian matrices,
have been significantly improved 
recently~\cite{DruskinKnizh98,MoretNovati04,EshofHochbruck06,PhD_Guettel} 
and successfully applied to solving time-dependent Maxwell
and wave 
equations~\cite{DruskinRemis2013,BornerErnstGuettel2015,%
Hochbruck_Pazur_ea2015,Botchev2016}.

We discuss implementation of our Krylov subspace exponential
integration method in detail in Section~\ref{sect:Kr1}. Unfortunately, this approach 
appears to have its limitations. To circumvent the limitations, we also design 
other solution strategies with the Krylov subspace exponential time 
integration. A key idea which leads to a very competitive method is to utilize 
the time asymptotic behavior of the solution in the Krylov subspace framework.
Another approach we develop is based on a splitting of the problem 
into a number of easier to solve homogeneous problems and nonhomogeneous
problems which appear to be identical due to the periodicity of the source term.

The rest of this article is organized as follows.
A precise formulation of the problem is discussed in 
Section~\ref{sect:problem}. The existing and proposed solution methods are 
presented in 
Section~\ref{sect:solvers}. We discuss numerical experiments in 
Section~\ref{sect:experim}, which is followed by conclusions.

\section{Problem formulation and current solution methods}
\label{sect:problem}
\subsection{Modeling multi-frequency optical response in
scattering medium}
We consider the scattering material to compose of infinitely long cylinders of 
radius $r$ which are randomly spaced with the minimum distance $d$ = $2r$. The 
symmetry of the scattering medium along $z$-axis (i.e., along the longitudinal 
axis of the cylinders) can be used to solve a two-dimensional computational 
model. The scattering material extends from x = $a_x$ to x = $b_x$ with perfect 
electric conductors at y = $a_y$ to y = $b_y$. The incident light originates 
from a point electric dipole. The propagation
of light is modeled by solving time-dependent Maxwell 
equations using finite-difference time-domain methods.  After a relatively long 
time, the light transmits through the scattering layer, where its intensity 
eventually approaches a
time asymptotic regime, i.e., becomes a time 
periodic function.  The result of the computations is the 
intensity speckle pattern formed in the transmission of the scattering layer.
The aim of the modeling is to study correlations of the speckle
patterns depending on the frequency $\omega$ (to be precisely defined 
later in this section) of the incident light.

\added{To numerically solve the Maxwell equations, we first make the equations 
dimensionless.  We follow
the standard dimensionless procedure as described, e.g., in
\cite{Botchev2016}.}
In the two-dimensional computational model, the unknown field components are 
$H_x(x,y)$, $H_y(x,y)$ (magnetic field) and
$E_z(x,y)$ (electric field) as the incident light originates from a point 
electric dipole.  Although the problem is two-dimensional,
the size of the discretized problem has to be large (in this paper
up to $\approx 5\cdot 10^6$ degrees of freedom) 
to resolve the inhomogeneities in the scattering medium.
More precisely, we solve the two-dimensional Maxwell
equations in a domain $(x,y)\in[a_x,b_y]\times[a_y,b_y]$ with 
perfectly electric conducting boundary conditions 
($E_z=0$) at the $y$-boundaries and nonreflecting
boundary conditions at the $x$-boundaries.
The nonreflecting boundary conditions are implemented 
as the perfectly matched layers (PMLs) and the \added{resulting 
dimensionless} Maxwell
equations read:
\begin{equation}
\label{mxw}
\begin{aligned}
\frac{\partial H_{x}}{\partial t} & 
= -\frac{1}{\mu_{r}} 
\frac{\partial E_{z}}{\partial y}, 
\\
\frac{\partial H_{y}}{\partial t} & 
= \frac{1}{\mu_{r}} \frac{\partial E_{z}}{\partial x} - \sigma_{x} H_{y} , 
\\
\frac{\partial E_{z}}{\partial t} & = \frac{1}{\epsilon_{r}} 
\left( \frac{\partial H_{y}}{\partial x} - \frac{\partial H_{x}}{\partial y}  
- J_{z}\right) - \sigma_x E_{z} + P,
\\
\frac{\partial P}{\partial t} & = -\dfrac{\sigma_x}{\epsilon_r}
\frac{\partial H_{x}}{\partial y},
\end{aligned}
\end{equation}
where $H_x$, $H_y$, $E_z$ are the unknown field components,
\deleted{$\mu_0$ and $\epsilon_0$ are respectively the magnetic permeability and 
electric permittivity of the vacuum, and}
$\mu_r$ and $\epsilon_r$ are the relative permeability and 
relative permittivity, respectively\added{, and
$P=P(x,y,t)$ is an auxiliary PML variable}.  \added{The details of the derivation of
the PML boundary conditions can be found in~\cite{Johnson2010_PML,Hanse2017_MSc}.}
In our case $\mu_r\equiv 1$ throughout the 
domain of interest and $\epsilon_r(x,y)$ is a piecewise constant
function representing the inhomogeneities in the scattering medium.  
Furthermore, the damping terms \replaced{$\sigma_x H_y$}{$\sigma_m H_{x,y}$}
(which \replaced{is}{are} nonphysical) and $\sigma_x E_z$ are
due to the PML boundary conditions. $\sigma_{x}$ \replaced{is}{are}
nonzero only in the so-called PML regions (situated just outside
of the domain $[a_x,b_x]\times[a_y,b_y]$ along the $x$-boundaries,
where the field is damped).  
Finally,
\begin{equation}
\label{src}
J_z(x,y,t)=\alpha(t) J(x,y), \qquad \alpha(t)=\sin(2\pi\omega t)
\end{equation}
is the source term, 
where $J(x,y)$ is nonzero only at the boundary $x=a_x$.
At the initial time $t=0$ initial conditions 
\begin{equation}
\label{ic1}  
H_x = 0, \quad H_y = 0, \quad E_z = 0\added{,\quad P=0}
\end{equation}
are imposed.
\deleted{In addition to equations~\eqref{mxw}, we have auxiliary
PML equations governing the conductivity values $\sigma_m(x,y)$
and $\sigma_e(x,y)$, see e.g.~\cite{Johnson2010_PML} for details.}

To solve~\eqref{mxw} with additional PML equations numerically, 
we \deleted{first make the equations dimensionless (with footnote: ``We follow
the standard dimensionless procedure as described, e.g., in
\cite{Botchev2016}.'')
and then} follow the method of lines approach, i.e., we first discretize
the equations in space.  In this paper, we use the standard 
finite-difference Yee space discretization, where the 
electric field values are situated at the grid vortices and 
the magnetic field values at the centers of the grid edges.
Alternatively, any other suitable space discretization can
be used for this problem, for instance, vector N\'ed\'elec finite 
elements (see e.g.~\cite{BotchevVerwer09})
or discontinuous Galerkin finite elements (see e.g.~\cite{Sarmany_ea2013}).
The space discretization then results in a system of ordinary
differential equations (ODEs)
\begin{equation}
\label{ode}
\left\{\begin{aligned}
y'(t) &= -A y(t) + \alpha(t) g,
\\
y(0)  &= v,
\end{aligned}\right.
\qquad
y = 
\begin{bmatrix}
\bar{H}_x \\ \bar{H}_y \\ \bar{E}_z \\ \bar{P}  
\end{bmatrix},
\quad
A = 
\begin{bmatrix}
\Ahat & B_1^T \\
-B_2    &  0  
\end{bmatrix},
\quad 0<t<T,
\end{equation}
where $\bar{H}_x$, $\bar{H}_y$, $\bar{E}_z$, \replaced{$\bar{P}$}{$y_{\textsc{pml}}$}
are the grid values of the unknown fields and the auxiliary
PML variable\deleted{s} and 
$\Ahat$ is the Maxwell operator matrix corresponding
to the space discretized equations~\eqref{mxw}, 
$$
\Ahat = \begin{bmatrix}
M_\mu^{-1}  & 0 \\
0         & M_\epsilon^{-1}
\end{bmatrix}
\begin{bmatrix}
M_{\sigma_x}  &  K \\
- K^T       &  M_{\sigma_x}
\end{bmatrix}.
$$
Here $M_{\epsilon,\mu,\sigma_{x}}$ are diagonal matrices containing
the grid values of $\epsilon$, $\mu$, $\sigma_{x}$, respectively,
and $K$ is a discretized two-dimensional curl operator.
In our problem, due to~\eqref{ic1}, the initial value $v$ is zero
but, for the purpose of presentation, we prefer to write $v$ there instead 
of a zero vector.

In the remainder of the paper we denote the total dimension of the 
problem~\eqref{ode} by $n$, i.e., $A\in\Rr^{n\times n}$.  
It can be shown that the eigenvalues of $\Ahat$ have nonnegative 
real part, see e.g.~\cite{BotchevVerwer09}.
We assume that the property holds for the matrix~$A$,
as is shown numerically in~\cite{Botchev2016}.

\section{Solution methods}
\label{sect:solvers}

\subsection{Standard finite difference time domain methods (FDTD)}
A well established and widely used class of methods to 
model electromagnetic scattering~\eqref{ode} is the finite-difference 
time-domain
methods~\cite{Taflove}.  In the framework of the method of lines,
these methods essentially imply a finite difference approximation
in space (often employing the Yee cell~\cite{Yee66}) and 
a subsequent application of a time integration method.
The celebrated Yee scheme is an example of this approach,
where the time discretization is based on a second order
symplectic composition scheme~\cite{BotchevVerwer09}.
This compact, energy preserving time integration scheme can 
be viewed as a combination of the staggered leap-frog scheme for the wave
terms (represented by the matrices $K$ and $K^T$ in~\eqref{ode})
and the implicit trapezoidal scheme (ITR) for the damping
terms $M_{\sigma_{x}}$.  However, when applied to the problems
with PML boundary conditions, the method loses its clarity
and represents essentially an ad-hoc splitting implicit-explicit scheme 
where the stiff PML terms are treated implicitly.

Our problem is two-dimensional and, hence, although large, can 
be efficiently treated by fully implicit FDTD schemes.  The arising
linear systems can then be solved by sparse direct solvers,
which have been significantly improved last 
decades~\cite{UMFPACK1,UMFPACK2,Duff_ea2017book}.
In addition, an employment of an implicit scheme is simple 
and removes the necessity to handle the PML terms in
a special way.  An implicit scheme which is very suitable
for the Maxwell equations is the classical ITR 
(also known as the Crank--Nicolson scheme).  
The scheme is second-order accurate, does not introduce artificial
damping and preserves energy~\cite{VerwerBotchev09}.
Therefore, in this paper we choose ITR to be the reference FDTD 
method. 
For our problem~\eqref{ode}
it reads
$$
\frac{y^{k+1} - y^{k}}{\tau} = 
-\frac12 A y^k -\frac12 A y^{k+1} 
+ \frac12 {\alpha^k}g  
+ \frac12 {\alpha^{k+1}}g,
$$
where $\tau>0$ is the time step size and the superindex $k$
indicates the time step number.
At every time step a linear system with a matrix
$I+\frac12\tau A$ has to be solved.

\subsection{Krylov subspace methods, basic facts}
\label{sect:Kr1}
Let $y_0(t)$ be an approximation (initial guess) to the solution
$y(t)$ of~\eqref{ode}.
Let us define the error (unknown in practice) and the 
residual of $y_0(t)$ as 
(cf.~\cite{CelledoniMoret97,DruskinGreenbaumKnizhnerman98,%
KnizhnermanSimoncini09})
\begin{equation}
\label{err_res}
\begin{aligned}
\eps_0(t) & := y(t) - y_0(t),
\\
r_0(t) & := -Ay_0(t) - y_0'(t) + \alpha(t) g.
\end{aligned}
\end{equation}
We assume that the residual $r_0(t)$ of the initial guess $y_0(t)$ 
is known.  This can be 
achieved, for instance, by taking $y_0(t)$ to be a constant 
function equal to the initial value $v$:
\begin{equation}
\label{rank1}  
r_0(t) = -Av + \alpha(t) g = \alpha(t) g,
\end{equation}
where the last step obviously holds only if $v=0$.
Note that the initial residual turns out to be 
a time dependent scalar function times a constant vector.
Furthermore, if $y_0(t)$ satisfies the initial condition $y_0(0)=v$,
we have
\begin{equation}
\label{corr1}
\eps_0'(t) = -A\eps_0 + r_0(t), \quad \eps_0(0) = 0.
\end{equation}
Starting with $y_0(t)$, Krylov subspace methods obtain 
a better solution $y_m(t)$ by solving~\eqref{corr1}
approximately:
\begin{equation}
\label{upd1}
y_m(t) = y_0(t) + \tilde{\eps}_0(t).
\end{equation}
Here $\tilde{\eps}_0(t)\approx\eps_0(t)$ is the approximate solution of~\eqref{corr1}
obtained by $m$ Krylov steps.
In this paper we use a regular Krylov 
subspace method~\cite{SaadBook2003,Henk:book}
as well as a rational shift-and-invert (SaI) Krylov subspace 
method~\cite{MoretNovati04,EshofHochbruck06}.  The two methods employ the
Arnoldi process to produce, after $m$~steps, 
upper-Hessenberg matrices $\Hh_m$, $\Hht_m$,
respectively, and an $n\times m$ matrix  
$$
V_m = 
\begin{bmatrix}
v_1 & \dots & v_m  
\end{bmatrix}\in\Rr^{n\times m},
$$
such that $V_m^*V_m$ is the $m\times m$ identity matrix and 
\begin{equation}
\label{Arn}
\begin{aligned}
\text{either}\quad AV_m &= V_{m+1}\Hh_m   \quad\text{(for regular Krylov 
method)},
\\
\text{or}\quad (I+\gamma A)^{-1}V_m &= V_{m+1}\Hht_m   
\quad\text{(for SaI Krylov method)}.
\end{aligned}
\end{equation}
Here $\gamma>0$ is a parameter whose choice is discussed later.
Note that for simplicity of presentation we slightly abuse notation in these 
last 
two relations: the matrix $V_{m+1}$ produced by the regular Krylov method
and appearing in the former relation
is different from $V_{m+1}$ 
produced by the SaI Krylov method and appearing in the latter relation.
Precise definition of the Arnoldi process is not given here as 
it can be found in many books, e.g., 
in~\cite[Algorithm~6.1]{SaadBook2003} or~\cite[Figure~3.1]{Henk:book}.
  
If the regular Krylov method is used, the columns of the matrix $V_m$ 
span the Krylov subspace 
$\Kk_m(A,g) \equiv\Span (g, Ag, A^2g, \dots, A^{m-1} g )$, i.e., 
\begin{equation}
\label{Krylov}
\Span(v_1,\dots,v_m)=\Kk_m(A,g), \qquad v_1=g/\|g\|  
\end{equation}
with $v_1$ being the normalized stationary part of the initial 
residual~\eqref{rank1}.  For the SaI method, the relations above hold with $A$ 
replaced by $(I+\gamma A)^{-1}$.
The SaI transformation results in a better approximation
in the Krylov subspace of the eigenmodes corresponding to
the small in modulus eigenvalues~\cite{MoretNovati04,EshofHochbruck06} 
and, hence, is favorable
for solving the time dependent problem~\eqref{corr1}.
Indeed, for some classes of the discretized differential operators $A$\added{, such
as the discretizations of parabolic PDEs with a numerical range along the 
positive real axis,}
one can show that the convergence of the SaI methods is mesh 
independent~\cite{EshofHochbruck06,GoecklerGrimm2014}.
\added{Although these results can not be extended to wave-type equations in a straightforward
manner, a mesh independent convergence} is observed
in practice for the Maxwell equations with damping \added{in}~\cite{Botchev2016}.

For the SaI Krylov method we define the matrix $H_m$ 
as the inverse shift-and-invert transformation:
\begin{equation}
\label{saiHm}
H_m = \dfrac1\gamma (\tilde{H}_m^{-1} - I).
\end{equation}
The approximate Krylov solution of the correction 
system~\eqref{corr1} can then be written 
for both Krylov methods as
\begin{equation}
\label{corr2}
\tilde{\eps}_0(t) = V_m u(t),
\qquad
\left\{
\begin{aligned}
u'(t) &= - H_m u(t) + \alpha(t)\beta e_1,
\\  
u(0) &= 0,
\end{aligned}
\right.  
\end{equation}
where $e_1=[1,0,\dots,0]^T\in\Rr^m$, $\beta=\|g\|$ and the ODE system in $u(t)$ 
is merely a Galerkin projection of~\eqref{corr1}
onto the Krylov subspace.

The two relations in~\eqref{Arn} are called Arnoldi 
decompositions.
It is convenient to use them after re-writing in an equivalent form  
\begin{equation}
\label{Arn1}
\begin{gathered}
AV_m = V_m H_m + R_m,
\\
R_m =
\begin{cases}
h_{m+1,m}v_{m+1}e_m^T
\quad &\text{(for regular Krylov method)},
\\
-\dfrac{h_{m+1,m}}{\gamma}(I+\gamma A)v_{m+1}e_m^T\tilde{H}_m^{-1}
\quad &\text{(for SaI Krylov method)},
\end{cases}
\end{gathered}
\end{equation}
where $e_m=[0,\dots,0,1]^T\in\Rr^m$.
This last form of the Arnoldi decompositions 
emphasizes the fact that the Krylov subspace can
be seen, if $R_m$ is small in norm, as an approximate invariant
subspace of $A$.

In both regular and SaI Krylov subspace methods, we control the number
of Krylov iterative steps $m$ (which is also the dimension
of the Krylov subspace) by checking the residual 
$r_m(t)=-y_m'(t)-Ay_m(t)+\alpha(t)g$ 
of the approximation $y_m(t)$ in~\eqref{upd1}.
As stated 
in~\cite{CelledoniMoret97,DruskinGreenbaumKnizhnerman98,BGH13}, 
the residual $r_m(t)$ can easily be computed as follows.

\begin{lemma}\cite{BGH13}
In the regular and SaI Krylov subspace methods~\eqref{corr1}--\eqref{corr2}
the residual $r_m(t)=-y_m'(t)-Ay_m(t)+\alpha(t)g$ 
of the approximate solution $y_m(t)$ to system~\eqref{ode}  
satisfies
\begin{equation}
\label{resid}
r_m(t) = 
\begin{cases}
-h_{m+1,m} e_m^Tu(t) v_{m+1}
\quad &\text{(for regular Krylov method)},
\\  
\dfrac{h_{m+1,m}}{\gamma} e_m^T \tilde{H}_m^{-1} u(t) (I+\gamma A)v_{m+1}
\quad &\text{(for SaI Krylov method)}.
\end{cases}
\end{equation}
\end{lemma}

\begin{proof}
For a detailed proof and discussion we refer to~\cite{BGH13}.
However, the problem considered there
is slightly different than~\eqref{ode}; it is of the form~\eqref{ode}
with $g=0$ and $v\ne 0$.
Therefore, for completeness of the presentation we give
a short proof of~\eqref{resid} here.
Based on~\eqref{corr1},\eqref{upd1}, we see that
$$
r_m(t)=-y_m'(t)-Ay_m(t)+\alpha(t)g = 
-\tilde{\eps}_0'(t)-A\tilde{\eps}_0(t) + r_0(t).
$$
Substituting $\tilde{\eps}_0=V_mu(t)$ into this relation
using the Arnoldi decomposition~\eqref{Arn1}, we obtain
$$
r_m(t) = V_m(-u'(t)-H_mu(t) +\alpha(t)\beta e_1) + R_mu(t)=
R_m u(t).
$$
\end{proof}

The methods presented in this subsection up to this point are 
well known, see 
e.g.~\cite{CelledoniMoret97,DruskinGreenbaumKnizhnerman98,BGH13}.  
A first, rather simple but, nevertheless,
important conclusion which can be drawn from the presentation is 
as follows.

\begin{coroll}
\label{res:no_rst}
The regular and SaI Krylov subspace methods~\eqref{corr1}--\eqref{corr2}
for solving the multi-frequency optical response problem~\eqref{ode}
are fully independent of the source time component $\alpha(t)$.  

In other words, if the problem~\eqref{ode} has to be solved for
many different $\alpha(t)$, the matrices $V_m$ and either 
$\Hh_m$ or $\Hht_m$ can be computed once and used for all
$\alpha(t)$.  Only the small projected problem~\eqref{corr2}
has to be solved for each new $\alpha(t)$ again.
\end{coroll} 

\begin{proof}
It is enough to note that by construction the Krylov subspace matrices 
$V_m$, $\Hh_m$ and $\Hht_m$ do not depend on $\alpha(t)$.
\end{proof}

Corollary~\ref{res:no_rst} allows to significantly save 
computational work when solving~\eqref{ode} numerically.
However, for this problem we can expect that the Krylov 
methods in the current form will not be efficient.  This is
because the required simulation time $T$ is very large (typically,
several hundred time periods of $\alpha(t)$), so that the Krylov
dimension $m$ can become prohibitively large in practice.
This effect should expected to be more pronounced for the regular
Krylov method, as the SaI method is typically efficient
in the sense that the Krylov dimension required for its convergence
is often bounded.
Nevertheless, the bad expectations are confirmed 
in the numerical experiments for both regular and SaI Krylov methods.  
Therefore, in the next section we discuss ways to restart
the Krylov subspace methods.

\subsection{Krylov subspace methods with restarting}
\label{sect:Kr2}
A very large Krylov dimension $m$ slows down the Krylov method as
$m+1$ basis vectors $v_1$, \dots, $v_{m+1}$ have to
be stored and every new basis vector has to be orthogonalized
with respect to all previous vectors.
Typical approaches to cope with the slowing down in time integration problems
are restarting in time and restarting in
Krylov dimension.  In the first approach, 
the time interval of interest $[0,T]$ is divided 
in a number of shorter time intervals, where the problem~\eqref{ode}
is solved subsequently.  This approach is used, for instance,
in the elegant EXPOKIT code~\cite{EXPOKIT}.

If we implement the restarting-in-time approach in our setting 
with many different $\alpha(t)$, 
we see that the initial value vector $v$ is not zero
in the second and subsequent time subintervals.  For this reason
the initial residual is no longer of the form scalar function
times a constant vector (cf.~\eqref{rank1}).  Instead, the residual
can be shown to be of the form $Up(t)$ where $U\in\Rr^{n\times 2}$,
$p: \Rr\rightarrow\Rr^2$.
For such problems, exponential block 
Krylov methods~\cite{Botchev2013,FrommerLundSzyld2017}
can be applied.  
With this method, it is also possible to solve problems~\eqref{ode}
simultaneously for different $\alpha(t)$.
The idea is then to represent the residuals 
corresponding to all different $\alpha(t)$ in one common expression
$Up(t)$, where the number of columns $k$ in $U$ can be  hopefully kept 
restricted using the truncated SVD (singular value decomposition).
This approach is described in detail in~\cite{Hanse2017_MSc}.
Numerical results presented there show this approach 
inefficient due the growth of $k$.

Another way to keep the Krylov dimension $m$ restricted is 
restarting in Krylov dimension.  A number of strategies
for Krylov dimension restarting are 
developed~\cite{TalEzer2007,Afanasjew_ea08,Eiermann_ea2011,%
PhD_Niehoff,BGH13,GuettelFrommerSchweitzer2014}.  
Here we consider the residual based restarting~\cite{BGH13},
which is slightly different from the approach~\cite{PhD_Niehoff} 
demonstrated to work successfully for solving Maxwell's
equations~\cite{Botchev2016}.

This restarting approach is based on the observation that
the residual in the Krylov method preserves the form~\eqref{rank1} 
of the initial residual $r_0(t)$.  Indeed, relation~\eqref{resid}
shows that $r_m(t) = \alphah(t)\gh$ with
$$
\begin{aligned}
\alphah(t)=-h_{m+1,m} e_m^Tu(t), \quad
\gh =v_{m+1}
\quad &\text{(for regular Krylov method)},
\\  
\alphah(t)=\dfrac{h_{m+1,m}}{\gamma} e_m^T \tilde{H}_m^{-1} u(t),
\quad
\gh =(I+\gamma A)v_{m+1}
\quad &\text{(for SaI Krylov method)}.
\end{aligned}
$$
Hence, to restart after making $m$ Krylov steps, 
we carry the update~\eqref{upd1}, set
\begin{equation}
\label{rst}  
y_0(t):= y_m(t), \qquad r_0(t):= r_m(t)=\alphah(t)\gh,
\end{equation}
discard the computed matrices $V_{m+1}$, $\Hh_m$ (or $\Hht_m$)  
and start the Arnoldi process again, with the first Krylov
basis vector $v_1=\gh/\|\gh\|$.
The method then proceeds as given by 
relations~\eqref{corr1}--\eqref{corr2}, with
$\alpha(t)$ and $g$ replaced by  $\alphah(t)$ and $\gh$,
respectively.
After making another $m$~Krylov steps, we can restart again.
The following result holds.

\begin{coroll}
\label{res:rst}
The restarted regular and SaI Krylov subspace 
methods~\eqref{corr1}--\eqref{corr2},\eqref{rst}
for solving the multi-frequency optical response problem~\eqref{ode}
are fully independent of the source time component $\alpha(t)$.  

In other words, if the problem~\eqref{ode} has to be solved for
many different $\alpha(t)$, the matrices $V_m$ and either 
$\Hh_m$ or $\Hht_m$ can be computed, at each restart, once and 
used for all $\alpha(t)$.  Only the small projected 
problem~\eqref{corr2} has to be solved, at each restart and 
for each new $\alpha(t)$, again.  
\end{coroll}

\begin{proof}
Observe that, as Corollary~\ref{res:no_rst} states,
the vector $v_{m+1}$ in both the regular and SaI Krylov
methods is independent of $\alpha(t)$.  Hence, so is the
vector~$\gh$.  Therefore, the second restart starts 
with $r_0(t):=\alphah(t)\gh$, where only $\alphah(t)$
depends on $\alpha(t)$.
\end{proof}

Note that we have to implement this restarting algorithm 
in such a way that the first restart is made for all 
the functions $\alpha(t)$, then the second restart for all the 
functions $\alphah(t)$, etc.
Otherwise (i.e., if first all the restarts were made for the first $\alpha(t)$,
then all the restarts for the second $\alpha(t)$, etc.), we would 
need to keep all Krylov bases from all the restarts.
The algorithm involves sampling and storing, for each $\alpha(t)$, 
the scalar functions $\alphah(t)$ at the end of each restart.
We outline the algorithm in Figure~\ref{alg:rst}.

\begin{figure}
  \centerline{\begin{minipage}{0.8\linewidth}
      \begin{enumerate}
      \item[(0)] 
        Set $y(T):=0$.
      \item[(1)]
        Carry out $m$ steps of the Arnoldi algorithm, compute the 
        matrices $V_{m+1}$ and $H_m$ (as given by~\eqref{saiHm}).
        \item[(2)] For each frequency $\omega$:\\
          (a) solve the projected problem~\eqref{corr2}; \\
          (b) sample 
          for $t\in[0,T]$ and store 
          $\alpha(t):=\frac{h_{m+1,m}}{\gamma} e_m^T \tilde{H}_m^{-1} u(t)$.
          \\
          (c) stop computations for $\omega$ values for which 
          the residual norm $|\alpha(t)|\|(I+\gamma A)v_{m+1}\|$
          is small enough.
        \item[(3)] 
          Carry out the update~\eqref{upd1}: $y(T):=y(T) + V_m u(T)$, 
          set $g:=(I+\gamma A)v_{m+1}$.
          Go to step~(1).
      \end{enumerate}
  \end{minipage}}
  \caption{Restarted Krylov subspace 
    method~\eqref{corr1}--\eqref{corr2},\eqref{rst} for solving~\eqref{ode} 
    for different frequencies $\omega$ in the source term $g=\alpha(t)g$.
    The algorithm is given for the SaI variant of the method.}
  \label{alg:rst}
\end{figure}

\subsection{Using the periodicity of the source function}
\label{sec:Kr3}
In this section we discuss two ways to make the numerical solution
of~\eqref{ode} more efficient by exploiting the time periodicity
of the source function $\alpha(t)g$.

\subsubsection{Krylov subspace methods with asymptotic correction}
\label{sect:Kr3}
Recall that the eigenvalues of $A$ have nonnegative real part and
the time interval of interest $[0,T]$ is large in the sense that the 
output light has reached a time-periodic ``steady'' state at time $T$. 
Therefore, we may hope to improve the Krylov subspace approximations by
splitting off this time-periodic part of the solution.

\begin{coroll}
The solution $y(t)$ to the problem~\eqref{ode} can be written
as 
\begin{equation}
\label{split1}  
y(t) = \yt(t)\replaced{-}{+} \yh(t), \qquad 
\yt(t)=\Im(e^{i2\pi\omega t}\zstat),\quad \yh(t)=e^{-tA} \Im\zstat,
\end{equation}
where $\zstat = (A+i2\pi\omega I)^{-1}g$
and $\Im z$ denotes the imaginary part of a vector $z\in\Cc^n$.
\end{coroll}

\begin{proof}
Since the source term is $\alpha(t)g$, with $\alpha(t) = \sin(2\pi\omega t)$
(cf.~\eqref{ode},\eqref{src}), the variation of constants formula (see 
e.g.~\cite{HundsdorferVerwer:book}) allows us to write  
the exact solution of~\eqref{ode} as 
$$
y(t) = \underbrace{e^{-tA}v}_{=0}{} + \int_0^t e^{(s-t)A}\sin(2\pi\omega s) g\, ds,
$$
\added{where the first term drops out because $v=0$, cf.~\eqref{ic1}.
Let us consider a function}
$$
z(t) = \int_0^t e^{(s-t)A}e^{i2\pi\omega s} g\, ds, \qquad i^2 = -1,
$$ 
\added{introduced in such a way that $y(t)=\Im z(t)$.}
This function can be brought to the form
$$\begin{aligned}
z(t) &= e^{-tA} \left[\int_0^t e^{s(A+i2\pi\omega I)} \, ds \right] g =
e^{-tA} \left. (A+i2\pi\omega I)^{-1} \Bigl[e^{s(A+i2\pi\omega I)} 
\Bigr]\right|_0^t g =
\\
&= e^{-tA} (A+i2\pi\omega I)^{-1} \Bigl[e^{t(A+i2\pi\omega I)}-I \Bigr] g =
e^{-tA} \Bigl[e^{t(A+i2\pi\omega I)}-I \Bigr] (A+i2\pi\omega I)^{-1}g =
\\
&= \Bigl[e^{i2\pi\omega t}-e^{-tA} \Bigr] (A+i2\pi\omega I)^{-1}g =
e^{i2\pi\omega t}\zstat - e^{-tA}\zstat.
\end{aligned}$$
Hence, we arrive at~\eqref{split1}.
\end{proof}

Note that
the first term $\yt(t)$ in~\eqref{split1} can be identified
as the time-periodic part of the solution, while the second one
$\yh(t)$ solves a homogeneous initial-value problem 
\begin{equation}
\label{y_homo}
\yh(t)' = -A\yh, \qquad \yh(0) = \Im\zstat.  
\end{equation}
Furthermore, we note that $\yt(t)$ is easy to compute and that 
solving the homogeneous problem~\eqref{y_homo} is in general an
easier task than solving an initial value problem for 
the inhomogeneous ODE system~\eqref{ode}.  
Indeed, the former problem is equivalent to evaluating the matrix exponential 
action which is, in general, cheaper than solving an ODE system~\cite{19ways}.
This is why evaluating $y:=e^{A}b$ for a given vector $b$ is often a subtask
in modern time integrators~\cite{HochbruckOstermann2010}.
Moreover, the Krylov subspace methods are known to work
well for problems of the type~\eqref{y_homo}, and this has been
used in the literature, see e.g.~\cite{PARAEXP}.

Another argument in favor of solving~\eqref{ode} by applying 
the splitting~\eqref{split1} is that we hope that $\yh(t)$ should
decay asymptotically and, hence, the solution $y(t)$
for large times should be well approximated by $\yt(t)$.  This 
hope is confirmed in practice (see Section~\ref{sect:experim}).
The Krylov subspace method applied to solve~\eqref{y_homo}
can easily be restarted in time and at every restart the norm
of the initial value vector is expected to be smaller.
Hence, the residual should be smaller and the Krylov subspace 
methods require less steps to converge.

To get the solutions for the other $\alpha(t)$ we first note that
$\zstat$ can be found simultaneously for many frequencies~$\omega$,
see e.g.~\cite[Section~4.3.1]{Hanse2017_MSc}. 
We do not discuss this further in this paper because, as compared to 
other costs, solving a system for $\zstat$ is very cheap anyway.
Furthermore, assume a solution $\yh_{\omega}(t)$ is found for a certain
frequency $\omega$.  Usually, a certain frequency range should
be scanned which means that the next frequency of interest 
$\omega_{\mathrm{next}}$
is rather close to $\omega$.  We first compute $\yt(t)$
for the new value $\omega_{\mathrm{next}}$.  Note that we can find 
$\yh_{\omega_{\mathrm{next}}}(t)$ 
as $\yh_{\omega_{\mathrm{next}}}(t) = \yh_{\omega}(t) + w(t)$, where $w(t)$ 
solves 
the problem
\begin{equation}
\label{4diff}
w'(t)=-Aw(t), \qquad w(0) = \yh_{\omega_{\mathrm{next}}}(0) -\yh_{\omega}(0).
\end{equation}
The Krylov subspace methods applied to this problem is likely to
require less iterations than for solving~\eqref{y_homo}
provided that 
$\|\yh_{\omega_{\mathrm{next}}}(0) 
-\yh_{\omega}(0)\|<\|\yh_{\omega_{\mathrm{next}}}(0)\|$.

\begin{figure}
\centerline{
  \begin{minipage}{0.8\linewidth}
For $\omega=\omega_1$:
\begin{enumerate}
  \item Compute $\zstat$\added{, $\yt(T)$} as given by~\eqref{split1}.
  \item Solve~\eqref{y_homo}, find~$\yh(t)$.
  \item Store 
    $\yh_\omega(0):= \Im\zstat$ and $\yh_\omega(T):=\yh(T)$.\\
    Compute the solution $y(T)=\yt(T) \replaced{-}{+} \yh(T)$.
\end{enumerate}
  For $\omega_{\mathrm{next}}$=$\omega_2$, $\omega_3$, \dots:
\begin{enumerate}
\setcounter{enumi}{3}
  \item Compute $\zstat$\added{, $\yt(T)$} as given by~\eqref{split1}, set
        $\yh_{\omega_\mathrm{next}}(0):=\Im \zstat$.
  \item Solve~\eqref{4diff}, find~$w(T)$, set 
    $\yh_{\omega_\mathrm{next}}(T) := \yh_{\omega}(T) + w(T)$.
  \item
    Update 
    $\yh_\omega(0):=\yh_{\omega_\mathrm{next}}(0)$ and 
$\yh_\omega(T):=\yh_{\omega_\mathrm{next}}(T)$.
    \\
    Compute the solution $y(T)=\yt(T) \replaced{-}{+} \yh_{\omega_\mathrm{next}}(T)$.
\end{enumerate}
  \end{minipage}}  
  \caption{Krylov subspace method based on the asymptotic 
splitting~\eqref{split1}}
  \label{alg:spl1}
\end{figure}

The algorithm for the Krylov subspace method with asymptotic splitting
is outlined in Figure~\ref{alg:spl1}.
It is important to note that solution of the homogeneous 
ODE systems~\eqref{y_homo} and~\eqref{4diff} at steps~2 and~4,
respectively, can be carried out with any restarting in time 
and in Krylov dimension.  This freedom is used by us 
to obtain an efficient time integrator in Section~\ref{sect:experim}.

\subsubsection{Splitting off the source term}
The second approach we consider here to exploit the time periodicity 
of the source function $\alpha(t)g$ in~\eqref{ode} is as follows.
We solve the problem~\eqref{ode} successfully on subintervals
$[0,\Delta T]$, $[\Delta T, 2\Delta]$, \dots{} (in other words,
we apply restarting in time).  To solve~\eqref{ode} on each
subinterval~$[(k-1)\Delta T,k\Delta T]$, $k=1,2,\dots$, 
we split the problem into two subproblems:
\begin{equation}
\label{split2}
\left\{
\begin{aligned}
w_1'(t) &= -Aw_1, 
\\
w_1(0)  &= y((k-1)\Delta T),
\end{aligned}
\right.  
\qquad\qquad
\left\{
\begin{aligned}
w_2'(t) &= -Aw_2 +\alpha(t)g, 
\\
w_2(0)  &= 0,
\end{aligned}
\right.  
\end{equation}
so that the solution for $t=k\Delta T$ is determined as 
$y(k\Delta T) = w_1(k\Delta T) + w_2(k\Delta T)$.
Note that the second problem has
the same solution for every subinterval
if $\alpha(t)$ is periodic (which is the case for our
problem) and we choose $\Delta T$ to be a multiple the time
period.  Thus the problem for~$w_2(t)$ has to be solved
for the first subinterval only. For the other subintervals,
it suffices only to solve the problem for~$w_1(t)$.  
This approach seems profitable because we again 
only have to solve a homogeneous problem, i.e., the matrix
exponential action has to be computed. 

\section{Numerical experiments}
\label{sect:experim}
\subsection{Test problem and implementation details}
Here we briefly discuss some implementation aspects important for 
successful application of Krylov subspace methods.
In all the experiments we use the SaI version of the Krylov
subspace method.  The regular Krylov subspace method appears to
be inefficient for this problem as too many Krylov iterations
are needed for its convergence.

All the numerical experiments presented in this paper are carried
out in Matlab on Linux PC with 8~CPUs and 32~Gb of memory.
To solve the linear systems with the matrices $(I+\frac\tau2 A)^{-1}$
in the ITR scheme  
and $(I+\gamma A)^{-1}$ in the SaI Krylov method,
Matlab's sparse direct LU~factorization (based on the 
UMFPACK~\cite{UMFPACK2}) is used.
The factorization can be carried out efficiently because the
problem is two-dimensional.
The factorization is computed once and used every time the action
of the inverse matrices is needed.

We consider the test case where the domain is
$$
(x,y)\in [-3,33]\times [0,10],
$$
and at $y=0$ and $y=10$ the perfectly conducting
boundary conditions are posed.  At $x=0$ and $x=30$
the PML boundary conditions are posed, with the PML
regions being $x\in[-3,-2]$ and $x\in[32,33]$.
The values of $\sigma_{x}$ grow in the PML regions 
from $0$ to $\sigma_{\max}=20$ quadratically. 
The final time is set to be $T=500$.

The region $[0,30]\times [0,10]$ contains $750$~cylinders 
of radius $0.1$ which are scattered randomly with
a minimum distance $0.2$ to the domain boundaries and 
$0.25$ from each other.  The electric permittivity
values are  
$\epsilon_r=1$ inside the cylinders and $\epsilon_r=2.25$ in the rest
of the domain (see Figure~\ref{fig:dom}).

\begin{figure}
\includegraphics[width=\textwidth]{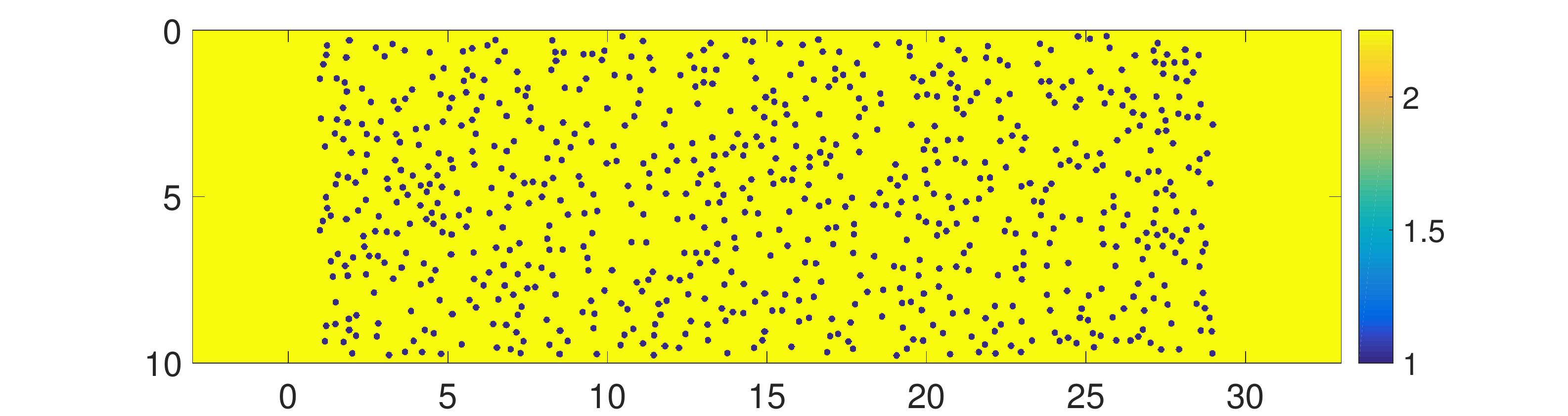}
\caption{The spatial domain is a waveguide with $750$ randomly positioned 
infinite cylinders}  
\label{fig:dom}
\end{figure}

The vector $g$ in the source function $\alpha(t)g$ is zero
everywhere in the domain except at the line $x=-2$, $y\in[0,10]$
where it is~$1$ for $y\in[1,9]$ and it decays linearly from 
$1$ to $0$ for $y\in[0,1]$ and $y\in[9,10]$.  The time component
$\alpha(t)$ of the source function is defined as 
$\alpha(t)=\sin2\pi\omega t$ with $\omega\in[0.85,1.15]$.

The highest resolution which can be used for this domain size 
is 64 grid points per unit length ($4\,498\,307$ , which means that every 
cylinder
is represented by approximately $12\times 12$ grid points.
Although this rather rough resolution is sufficient for simulation
purposes, to have consistent results for all mesh resolutions
we regularize the values of $\epsilon_r$ by applying
a standard homogenization procedure as follows.  Let 
$(\epsilon_r^0)_{i,j}$ represent the original piecewise 
constant values of $\epsilon_r$  at the mesh points $(i,j)$
for the mesh resolution~256 points per unit length (this 
resolution is too high to be used for the whole simulation). 
Then, we update 
$$
(\epsilon_r^{k+1})_{i,j} = 
\frac12 (\epsilon_r^k)_{i,j} +
\frac12 \frac{(\epsilon_r^k)_{i-1,j} + (\epsilon_r^k)_{i,j-1} + 
              (\epsilon_r^k)_{i,j+1} + (\epsilon_r^k)_{i+1,j}}4
$$
iteratively for $k=0,1,2,\dots$ until the iterations stagnate
(at $k\approx 200$).
The resulting values of $\epsilon_r$ are then interpolated
onto the coarser meshes used in simulations.
\replaced{Similar homogenization procedures are also employed in FDTD codes, such as
MEEP~\cite{Oskooi_ea2009meep}.}{For another type of homogenization used in an FDTD code 
Meep see~\cite{Oskooi_ea2009meep}.}

The matrix $H_m$ produced by the regular Krylov method is 
a Galerkin projection of the matrix $A$, which means that
$H_m=V_m^T A V_m$ and the field of values of $H_m$ is
a subset of the field of values of $A$.
However, for the SaI Krylov method the matrix $H_m$
results from the inverse SaI transformation~\eqref{saiHm}
and therefore can have spurious eigenvalues.
This can be especially pronounced for the matrices $A$ with
purely imaginary eigenvalues (or eigenvalues with very small
real part), as is the case for our problem~\eqref{ode}, 
see e.g.~\cite{Botchev2016}.
Indeed, if $A$ has a purely imaginary eigenvalue $\lambda=iy$, $y\in\Rr$,
then $\tilde{H}_m$ can have eigenvalues \emph{approximating}
$\tilde{\lambda}=(1+\gamma\lambda)^{-1}=(1+i\gamma y)^{-1}$.  
The inverse SaI transformation of this approximate $\tilde{\lambda}$
may have a small real part of a negative sign, especially for large
$\gamma y$.
In practice, this spurious eigenvalues with a small real part of the wrong
sign can be replaced by the eigenvalues with zero real part.
More precisely, once $H_m$ is computed according to~\eqref{saiHm},
we compute its Schur decomposition $H_m=QTQ^*$ and replace
the small negative diagonal entries in $T$ (if any) by zero.
After that, we set $H_m=QTQ^*$ where $T$ is now the corrected
matrix.

We choose the parameter $\gamma$ in the SaI Krylov subspace method
by making cheap trial
runs on a coarse mesh (in this case with resolution $16$~points
per unit length) and using the chosen $\gamma$ for all the
meshes~\cite{Botchev2016}.
This is possible because the SaI methods usually exhibit 
a mesh independent 
convergence~\cite{EshofHochbruck06,GoecklerGrimm2014,Botchev2016}.
We also look at the number and size of the spurious eigenvalues
in the back SaI transformed matrix $H_m$.  
Typically we see that too many wrong eigenvalues can 
appear for large $\gamma$, so that $\gamma$ should be chosen 
not too large.
Table~\ref{t1} shows
the data of the test runs carried out to choose $\gamma$.  Based 
on the data, for this particular case, we set $\gamma:=0.0005 T = 0.25$.

\begin{table}
\caption{The residual norm $\|r_m(T)\|$
(as defined by~\eqref{resid})
for different values of $\gamma$ with resolution 16 and $\omega = 1$
and $m=400$ iterations. 
The $*$ sign means that some spurious eigenvalues have been detected and 
corrected, 
$**$ means that the value of $\gamma$ is unacceptable as 
there are too many or too large spurious eigenvalues.}
\label{t1}
\centering
\[
\begin{array}{|l|l|} \hline
\frac{\gamma}{T} & \text{residual norm} \\ \hline\hline
0.1$*$ & 1.63 \e{0} \\ \hline
0.05 & 1.60 \e{0} \\ \hline
0.025$*$ & 1.47 \e{0} \\ \hline
0.01 & 1.54 \e{0} \\ \hline
0.005 & 1.36 \e{0} \\ \hline
0.0025 & 1.39 \e{0} \\ \hline
0.001$*$ & 1.10 \e{0} \\ \hline
0.0005$*$ & 9.15 \e{-1} \\ \hline
0.00025$**$ & 7.60 \e{-5} \\ \hline
\end{array}
\]
\end{table}

To solve the projected ODE system~\eqref{corr2}, we use 
one of the standard Matlab's built-in stiff ODE solvers.
It is important to use a stiff solver due to the PML
boundary conditions and because the inverse SaI 
transformation~\eqref{saiHm} can yield large eigenvalues in $H_m$.
In case a homogeneous projected ODE system is solved,
i.e., $u'(t)=-H_mu(t)$, 
its solution is computed with the help of Matlab's matrix exponential
function \texttt{expm} (see e.g.~\cite{Higham_bookFM}).

\subsection{Numerical results and discussion}
We now compare the presented methods for the highest 
grid size possible on the PC we have, namely 64 grid points
per unit length.  On this mesh the size of the system~\eqref{ode} to
be solved is $n=4\,498\,307$.  Since the spatial error is expected
to be of order $(1/64)^2$, we should aim at the temporal error of
the same order.  Therefore, we set the residual tolerance in the
Krylov methods to $10^{-4}$.  Recall that for the ITR method there is no 
residual
available and, hence, its temporal error can not be easily controlled.
Due to wave character of the problem, we expect that the time step 
$\tau$ in ITR should be at least of the same order as the spatial grid step,
i.e., approximately $1/64$.  

To check the actual accuracy achieved
by the methods under comparison, we report the relative temporal error
computed with respect to a reference solution $y_{\text{ref}}(t)$, i.e.,
$\|y(T)-y_{\text{ref}}(T)\|/\|y_{\text{ref}}(T)\|$.  The reference
solution is computed by the standard FDTD method ITR with a tiny
time step size.  Thus, this reference solution is expected to
have the same spatial error as the methods to be compared.
Hence, $y(T)-y_{\text{ref}}(T)$ can effectively be seen as the temporal 
error.

We start with examining the performance of the ITR scheme,
see Table~\ref{t:ITR}.  Next, Table~\ref{t:Kr2} shows
the results for the Krylov subspace 
method~\eqref{corr1}--\eqref{corr2},\eqref{rst}
with restarting in dimension, as presented in Section~\ref{sect:Kr2}.
In the table, we also show the results for the much
coarser spatial mesh to demonstrate that the convergence
of the method does not deteriorate as the mesh gets
finer.  Due to the independency of the method on the
source time component $\alpha(t)$ (cf.~Corollary~\ref{res:rst}), 
the CPU time required by the method for any other frequency
from the range of interest is the CPU time needed for the
projected ODE system, i.e., $6.21 \e{3}$~s.
Thus, for the second and subsequent frequencies, 
the gain we obtain with respect to the ITR scheme is
a factor of $1.32 \e{5} /6.21 \e{3}\approx 20$.
A drawback of this method is that the used restart value $m_{\max}=400$ is 
very large, a larger restart value would hardly be
possible due to computer memory limitations.  For 
$m_{\max}=200$ no convergence is observed.

\begin{table}
\caption{Performance of the ITR scheme}
\label{t:ITR}
\centering
\[ \begin{array}{|l|l|l|} \hline
\tau & \text{relative error} & \text{CPU time} \\ \hline\hline
\tau=\Delta x = 1/64 & 4.12 \e{-1} & 1.84 \e{4} \\ \hline
\tau=\Delta x/2      & 1.07 \e{-1} & 3.74 \e{4} \\ \hline
\tau=\Delta x/8      & 6.33 \e{-3} & 1.32 \e{5} \\ \hline
\end{array} \]
\end{table}

\begin{table}
\caption{Results for 
the Krylov subspace method~\eqref{corr1}--\eqref{corr2},\eqref{rst}
restarted in dimension every $m_{\max}=400$ Krylov steps.}
\label{t:Kr2}
\centering
\[
\begin{array}{|l|l|l|l|l|l|} \hline
\text{resolution} &  \# \text{restarts} & \text{residual} & \text{relative} & 
\text{total CPU} & \text{CPU time for}  \\
 &  & \text{norm} & \text{relative} & \text{time} & \text{projected ODE}  \\ 
\hline\hline
16 & 14 & 4.88 \e{-5} & - & 2334.5 & 1252.0 \\ \hline
64 & 12 & 1.53 \e{-5} & 5.76 \e{-3} & 3.09 \e{4} & 6.21 \e{3} \\ \hline
\end{array}
\]
\end{table}

We now present the results for the two methods based on the 
splittings~\eqref{split1} and~\eqref{split2}, respectively.
In this methods, we are free to use both restarting in time
and in space.  By making cheap trial runs on coarse meshes,
we choose the subinterval length for restart in time to be~$1$.
After that, the parameter $\gamma$ is chosen as explained
above, resulting in the value $\gamma=0.01$.  With these 
parameters, the Krylov subspace dimension has not 
exceeded~$25$ throughout restarts.

The results are given in Table~\ref{t:Kr3}.
First of all, we demonstrate there that the time-periodic asymptotic solution 
$\yt(t)$ is a good approximation to the solution $y(t)$,
it is even more accurate than ITR with the time step $\tau=\Delta x$.
However, the accuracy of this solution is not of the order
of the spatial error, which may not be enough.
Comparing the results for the methods based on the 
splittings~\eqref{split1} and~\eqref{split2}, we see
that splitting~\eqref{split1} outperforms the other splitting.
For both splittings a homogeneous problem of the type~\eqref{y_homo}
has to be solved.  The difference in performance is because
the Krylov subspace method significantly profits from the fact that
$\yh(t)$ gets smaller in norm as time grows.
Indeed, a small \added{in} norm initial value means a small initial
residual (cf.~\eqref{err_res},\eqref{rank1} with $g=0$ and a small
in norm $v$) and, hence, less steps to satisfy the required
tolerance, see Figure~\ref{fig:normy}.
For the same reason, combination of this splitting
with~\eqref{4diff} turns out to be successful as well.

\begin{figure}
\includegraphics[width=0.5 \textwidth]{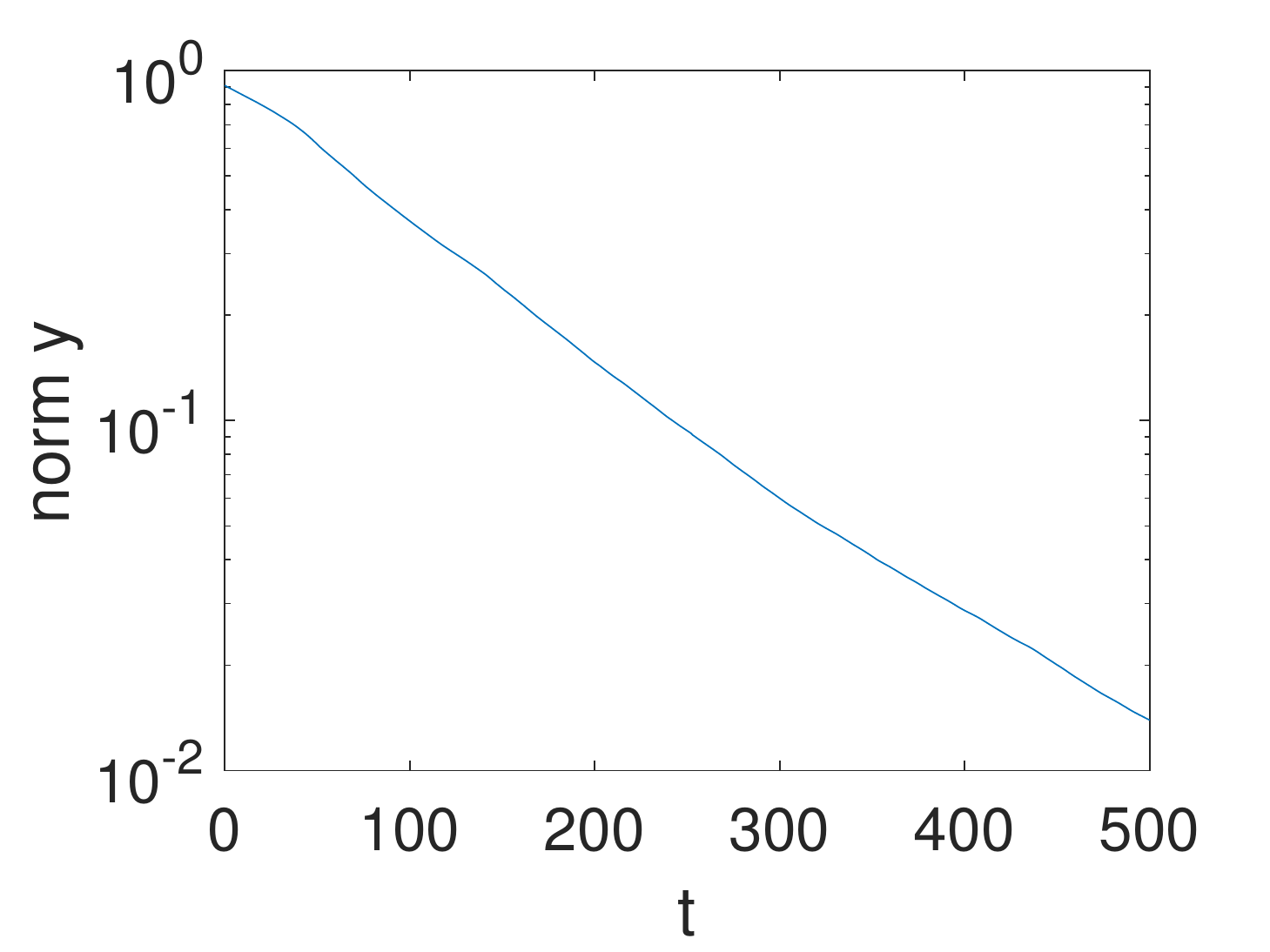}
\includegraphics[width=0.5 \textwidth]{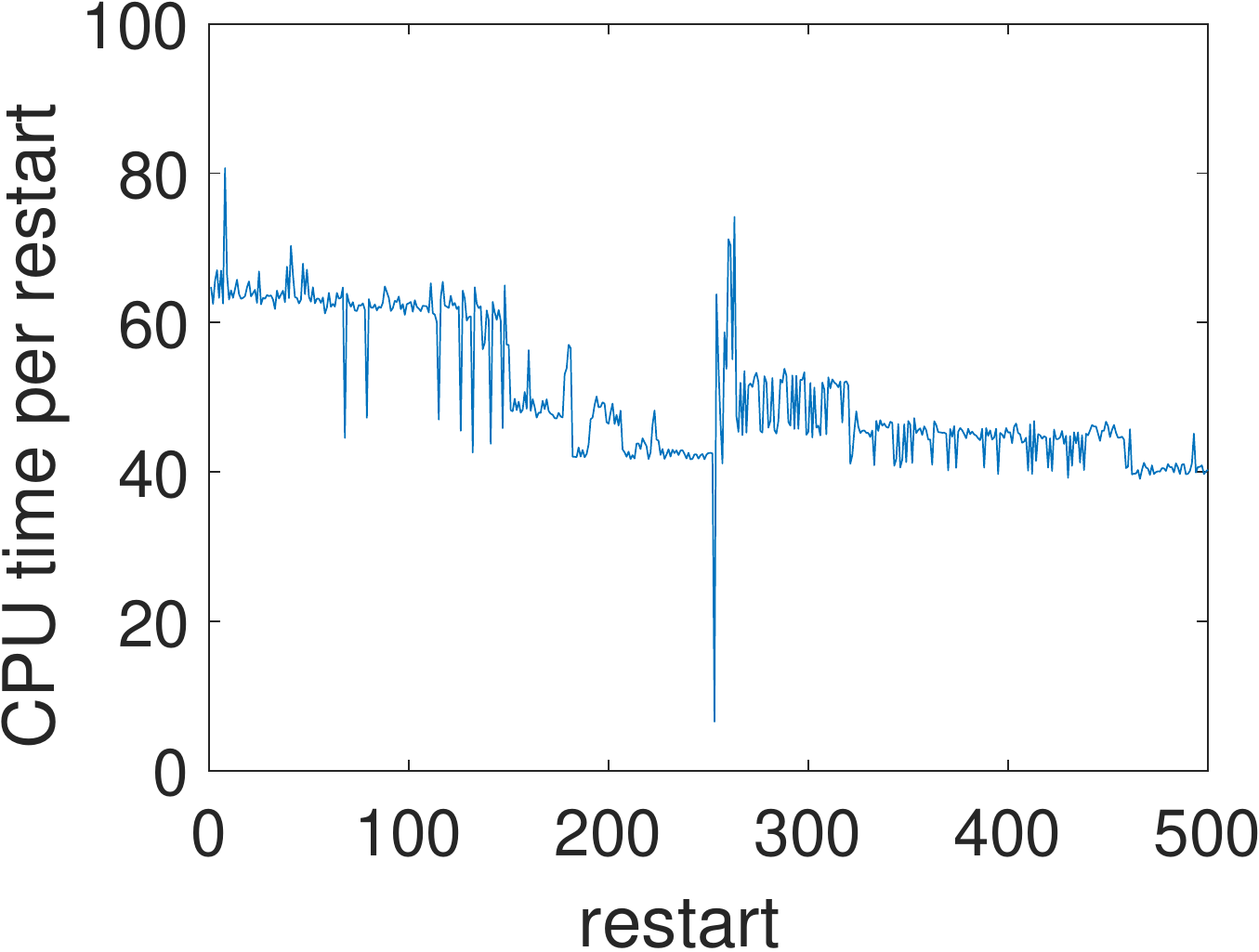}  
  \caption{Left plot: $\|\yh(t)\|$ versus $t$.  Right plot: corresponding
CPU times versus restarts in time.}
\label{fig:normy}
\end{figure}

As shown, the method based on the splitting of
the time-periodic asymptotic solution seems very promising:
a significant gain factor ($\sim$ 13) with respect to the ITR method is 
achieved without utilizing a lot of memory. 

\begin{table}
\caption{The performance of the two Krylov subspace solvers 
from Section~\ref{sect:Kr3}}
\label{t:Kr3}
\centering
\begin{equation*} 
\begin{array}{|l|l|l|} \hline
\text{method} & \text{relative error} & \text{CPU time} \\ \hline\hline
\text{time-periodic asymptotic solution $\yt(t)$} & 1.60 \e{-2} & 4.95 \e{1} \\ 
\hline
\text{splitting $\yt(t)+ \yh(t)$, cf.~\eqref{split1}} & 6.39 \e{-3} & 2.53 
\e{4} \\ \hline
\text{splitting $\yt(t)+ \yh(t)$, next $\omega_{\mathrm{next}}=\omega+0.001$} & 
- & 9.75 \e{3} \\ \hline
\text{splitting $w_1(t)+ w_2(t)$, cf.~\eqref{split2}} & 7.08 \e{-3} & 4.01 
\e{4} \\ \hline
\end{array}
\end{equation*}
\end{table}

Finally, in Table~\ref{t:fin} we collect the results for the most
promising methods and for the reference method ITR (run 
with the time step providing the required time accuracy 
$O(\Delta x)^2$).
As shown for a single frequency, the Krylov subspace method
based \added{on the} time-periodic asymptotic splitting is the fastest
and provides a gain factor of $5.2$ with respect to the 
reference method.  At other frequencies, its gain
factor of $13.5$ is lower that the gain factor of the 
Krylov subspace method with restarting in dimension ($21.3$).

\begin{table}
\caption{Results for the most promising Krylov subspace methods 
and for the reference method ITR (gathered from 
Tables~\ref{t:ITR}--\ref{t:Kr3})}
\label{t:fin}
\[
\begin{array}{|l|l|l|l|l|} \hline
\text{Method} & \text{Error} & \text{CPU time} & \text{Gain} & \text{Notes} 
\\
              &              &                 & \text{factor} & 
\\ \hline\hline
\text{ITR, }\tau=\Delta x/8 & 6.33 \e{-3} & 1.32 \e{5} & 1 & \text{reference 
method} \\ \hline
\text{Krylov, splitting~\eqref{split1}} & 6.39 \e{-3} & 2.53 \e{4} & 5.22 &  \\ 
\text{Krylov, splitting~\eqref{split1}}, \omega_{\mathrm{next}} & - & 9.75 
\e{3} & 13.52 & \omega-\omega_{\mathrm{next}}=0.001 \\ \hline
\text{Krylov, restarts in dimension} & 5.76 \e{-3} & 3.09 \e{4} & 4.27 & 
\text{high memory} \\ 
\text{Krylov, restarts in dimension}, \omega_{\mathrm{next}} & - & 6.21 \e{3} & 
21.26 & \text{requirements} \\ \hline
\end{array}
\]
\end{table}

\section{Conclusions}
Standard Krylov subspace methods turn out to be inefficient
in solving multi-frequency optical response from a scattering medium due to the 
growth of Krylov dimension. We overcome this inefficiency through two 
restarting strategies to restrict the Krylov dimension.
The first approach is to restart in time, i.e., to use Krylov
subspace methods on successive shorter 
subintervals.  
In this approach the invariance of the Krylov subspace on the
time component $\alpha(t)$ is lost and we would need to use a block
Krylov subspace.  This approach is worked out and demonstrated to be 
inefficient in~\cite{Hanse2017_MSc} due to the growth of the block
size.  

The other restarting approach we consider is the residual based
restarting in Krylov dimension.  This approach is shown to lead
to a method where the large scale part of the computational work 
does not depend on the time component $\alpha(t)$.  As numerical 
experiments demonstrate, the new method works successfully only 
for very large Krylov dimensions.  In the tests, this method appears to
be the fastest, at the cost of high memory consumption.

To avoid the high memory requirements, we consider two other
approaches based on the splitting.  In the first one, the solution
is split into an easy-to-compute time-periodic asymptotic part and 
the remaining part which decays with time.  The Krylov subspace methods
are demonstrated to be able to compute this decaying component very
efficiently, thus providing a rigorous numerical solution
to the whole problem.  In the second approach, the periodicity
of the source term is exploited.  The problem is solved
successfully on time subintervals and, on each subinterval,
it is split into a homogeneous ODE system (i.e., with zero source term)
and an inhomogeneous ODE whose solution is the same for all
subintervals.  The first, time-periodic asymptotic splitting
appears to be more efficient and works well for multiple
frequencies.

Thus, Krylov subspace exponential integrators seem
to be a promising computational tool for modeling 
multi-frequency optical response with monochromatic sources.

\bibliographystyle{abbrv}
\bibliography{my_bib,matfun,mxw,bibdata_Ravi}

\end{document}